\documentclass[12pt]{article}
\usepackage{amsfonts,amsmath}
\usepackage[mathscr]{eucal}
\usepackage{amssymb}
\usepackage{cite}
\usepackage{bm}
\usepackage{bbold}
\usepackage{amsthm}
\usepackage{graphicx}
\usepackage{epsf}
\theoremstyle{plain}
\newtheorem{thm}{Theorem}

\newtheorem{rem}{Remark}

\def\theequation{\arabic{section}.\arabic{equation}}
\newcommand{\be}{\begin{eqnarray}}
\newcommand{\ee}{\end{eqnarray}}
\newcommand{\nn}{\nonumber \\}
\newcommand{\lb}{\label}
\newcommand{\p}[1]{(\ref{#1})}

\newcommand{\ga}{\lower.7ex\hbox{$
\;\stackrel{\textstyle>}{\sim}\;$}}
\newcommand{\vecg}[1]{\mbox{\boldmath $#1$}}
\newcommand{\la}{\lower.7ex\hbox{$
\;\stackrel{\textstyle<}{\sim}\;$}}

\begin{document}

\begin{titlepage}

\vspace*{0.2cm}

\renewcommand{\thefootnote}{\star}
\begin{center}

{\LARGE\bf  Classical and Quantum Dynamics of Higher-Derivative Systems}\\

\vspace{0.5cm}

\vspace{1.5cm}
\renewcommand{\thefootnote}{$\star$}

\quad {\large\bf Andrei~Smilga} 
 \vspace{0.5cm}

{\it SUBATECH, Universit\'e de Nantes,}\\
{\it 4 rue Alfred Kastler, BP 20722, Nantes 44307, France;}\\
\vspace{0.1cm}

{\tt smilga@subatech.in2p3.fr}\\

\end{center}
\vspace{0.2cm} \vskip 0.6truecm \nopagebreak

\begin{abstract}
\noindent A brief review of the physics of systems including higher derivatives in the Lagrangian is given. All such systems involve {\it ghosts}, i.e. the spectrum of the Hamiltonian is not bounded from below and the vacuum ground state is absent.
Usually this leads to collapse and loss of unitarity. In certain special cases, this does not happen, however: ghosts are {\it benign}.

We speculate that the Theory of Everything is a higher-derivative field theory, characterized by the presence of such benign ghosts and defined in a higher-dimensional bulk. Our Universe represents then a classical solution in this theory, having the form
of a 3-brane embedded in the bulk. 
\end{abstract}

\vspace{1cm}
\bigskip

\newpage

\end{titlepage}

\setcounter{footnote}{0}

\setcounter{equation}0
\section{Motivation}
In the vast majority of problems in theoretical mechanics and theoretical physics, the Lagrangians
depend on generalized coordinates and generalized velocities, but not on generalized accelerations or still higher derivatives of dynamic variables. In particular, this is true for the Lagrangian of the Standard Model.

However, having unravelled by $\sim 1975$ the underlying structure of matter at the scale $\ga
10^{-17}$ cm, the theorists met a major impasse: the attempts to construct quantum theory of gravity have not been successful so far. There are two main reasons for that.
\begin{enumerate}
\item A direct attempt to quantize Eisnstein's gravity gives a nonrenormalizable quantum theory. Such a theory does not make sense not only because of the impossibility of perturbative calculations, but also due to impossibility to define path integrals: in a nonrenormalizable theory, the lattice path integral has no continuum limit.
\item In any gravity theory, time and spatial coordinates are intertwined so that the classical equations of motion do not represent a Cauchy problem --- the evolution in independent flat time. As a result, causality is lost. Einstein's equations admit strange G\"odel solutions with closed time loops \cite{Goedel}. Thus, even the classical general relativity has fundamental difficulties whose resolution is not in sight.
And things do not become more assuring when one tries to quantize it.
\end{enumerate}

Today the prevailing  opinion of the experts is that the Theory of Everything that includes quantum gravity is superstring theory of some kind. However, the latter has its own serious difficulties. 

\begin{enumerate}
\item Its building is impressive and beautiful, but the lowest stories of this building are hid in mist.
String theory is simply not formulated at the fundamental nonperturbative level.
\item More that 30 years have now passed since the superstring revolution of 1985, but string theory still cannot boast of phenomenological successes. It has not contributed much in our understanding of why the world we see is this and not that.
 \end{enumerate}
 
 Bearing all this in mind, it is interesting to pursue as far as we can an alternative, more conservative line of reasoning, assuming  that the TOE is not a string theory, but a conventional field theory living in flat
 space with universal flat time.
 
 \begin{itemize}
 \item The first question that one should be able to answer in this approach is how to explain
 the observed fact that the space-time we live in is curved. And the only way to do so is to assume
 that our (3+1)--dimensional Universe represents a thin curved film embedded in a flat higher-dimensional bulk, like a 2-dimensional soap bubble is embedded in flat 3-dimensional space.
 
 The fundamental TOE should be formulated in this bulk, and our Universe should represent a classical solution in this theory, a kind of Abrikosov string, but  extended not in one, but in three spatial directions, and very thin in the transverse directions so that the latter cannot be perceived.
 
 \item Then the question is how this higher-dimensional theory can look like. The problem is that we
 cannot write any Lagrangian, familiar from 4-dimensional physics, like
  \be
 {\cal L} \ =\ -\frac 1{2g^2} {\rm Tr} \{ F_{\mu\nu} F_{\mu\nu} \} \, .
   \ee
  We cannot do that because at $D>4$ the coupling constant $g$ acquires dimension, and the theory is no longer renormalizable. 
  
  What one can try to do is to add extra derivatives. For example, the 6-dimensional Lagrangian
   \be
   \lb{LD=6}
   {\cal L}^{D=6} \ = \ \alpha  {\rm Tr} \{ F_{\mu\nu} \Box F_{\mu\nu} \} + 
   \beta {\rm Tr} \{ F_{\mu\nu} F_{\nu\alpha}F_{\alpha\mu} \}
    \ee
    involves dimensionless coupling constants $\alpha, \beta$ and is renormalizable.
    
    \item We thus arrive at a higher-derivative (HD) theory, and that was my own main motivation to study such theories --- I tend to believe that the Holy Grail TOE is a HD theory living in a higher-dimensional bulk, with our Universe representing a curved 3-brane embedded there \cite{TOE}.
    \footnote{Unfortunately, we have now absolutely no idea {\it what} this theory might be.}
    
     But we should say right now that all such theories have a common striking feature --- they are
     {\it ghost-ridden}. This means that the quantum Hamiltonian of any such theory does not have a ground state and its spectrum involves the levels with arbitrarily low energies. Neither has this Hamiltonian a ``sky state'' ---  the spectrum also involves the levels with arbitrarily high energies.

 \end{itemize}  

For many years people thought that such theories are inherently  sick and did not study them. Indeed, the absence of the vacuum ground state is a very strange unusual feature. And in many cases such theories are sick, indeed. They involve a {\it collapse} leading to violation of unitarity. But it has become clear recently that a collapse and violation of unitarity is a feature of {\it many}, but not of
{\it all} ghost-ridden theories. There are quantum mechanical and field theory systems with ``benign
ghosts'' --- their spectrum is bottomless and there is no vacuum state, but still the quantum problem is well posed, there is no collapse and the evolution operator is unitary.

 Unfortunately, there have been much confusion in the literature devoted to this problem. But correct understanding of this  issue is now being gradually achieved. The main purpose of this mini-review is to clarify the remaining confusions and to explain what is known now about the physics of HD systems in simple terms.

 \section{Ostrogradsky
  Hamiltonian and its properties}
\setcounter{equation}0

 The first study of HD systems is due to a Russian mathematician Mikhail Ostrogradsky
 \footnote{He is known also by Ostrogradsky's theorem of vector analysis and by Ostrogradsky's method of calculating the integrals $\int dx \,  P(x)/Q(x)$ for the ratios of two polynomials.}  who developed back in 1848 the Hamiltonian method in classical mechanics independently of William Hamilton
 \footnote{Ostrogradsky did not read the splendid work of Hamilton of 1835 \cite{Hamilton} --- it was written in English and in those elder days it was French rather than English that played the role
 of the international language of science.}
 and applied it to the ordinary and also to HD dynamical systems \cite{Ostro}. 
 
 Let us explain what Ostrogradsky did, using the modern notation. Consider first the simplest nontrivial
 case when the Lagrangian of the system, 
  \be
  \lb{Lagrx2}
  L(x, \dot{x}. \ddot{x}) \, ,
  \ee
 depends on a single generalized coordinate $x$, the generalized velocity $v = \dot{x}$ and the generalized acceleration $a = \ddot{x}$. 
 The least action principle gives the equation of motion
  \be
  \lb{Lagr-eqmot-2}
  \frac {d^2}{dt^2} \left( \frac {\partial L}{\partial \ddot{x}} \right) - 
  \frac {d}{dt} \left( \frac {\partial L}{\partial \dot{x}} \right) + \frac {\partial L}{\partial x} \ =\ 0 \, .
   \ee
   The equation \p{Lagr-eqmot-2} is of order 4 and to solve it one has to specify four initial values
   $x(0), \dot{x}(0), \ddot{x}(0), x^{(3)}(0)$. This implies the existence of two pairs of canonical
   variables in the Hamiltonian description. 
   
   The Lagrangian \p{Lagrx2} does not depend explicitly on time, which entails energy conservation.
   Indeed, one can check that the quantity 
     \be
     \lb{energy}
     E  \ =\ \ddot{x}  \frac {\partial L}{\partial \ddot{x}} + \dot{x} \left(  \frac {\partial L}{\partial \dot{x} } -   \frac {d}{dt}  \frac {\partial L}{\partial \ddot{x}} \right) - L 
       \ee
       is an integral of motion, $dE/dt = 0$. 
       
       To construct the Hamiltonian, treat 
       \be
       \lb{vdef}
       v = \dot{x}
        \ee
        as an independent variable, define the canonical momenta
        \be
        \lb{mom-def-2}
        p_v \ =\ \frac {\partial L}{\partial \dot{v}} = \frac {\partial L}{\partial \ddot{x}} \, , \nn
        p_x \ =\ \frac {\partial L}{\partial \dot{x}} -  \dot{p}_v \, ,
          \ee
        express the generalized velocities via $p_v$ and $p_x$ and rewrite \p{energy} as
         \be
         \lb{Ostr-Ham-2}
         H(p_v, p_x; v, x) \  = \ p_v \dot{v} + p_x \dot{x} - L  \ =\ 
         p_v a(p_v, x, v) + p_x v - L \left[a(p_v, x, v), x, v \right] \, ,
           \ee
           where $a(p_v, x, v)$ is the solution of the equation $\partial L(x,v,a)/\partial a = p_v$. 
           The 
           Hamilton equations of motion are
           \be
           \lb{Ham-eqmot}
          I: \ \  \frac {\partial H}{\partial p_v} \ =\ \dot{v}, \ \ \ \ \ \ \  II: \ \  \frac {\partial H}{\partial v} \ =\ -\dot{p_v}, \nn
   III: \ \ \frac {\partial H}{\partial p_x} \ =\ \dot{x}, \ \ \ \ \ \ \  IV: \ \  \frac {\partial H}{\partial x} \ =\ -\dot{p_x}  \, .
    \ee
    We see that  equation $(III)$ coincides with \p{vdef},  equations $(I)$  and $(II)$ are fulfilled, bearing in mind the definition of the canonical momenta in \p{mom-def-2}, and equation $(IV)$ is equivalent, bearing
    in mind all these definitions, 
    to   the Lagrange equation of motion
    \p{Lagr-eqmot-2}.    
             
             If the Lagrangian depends on $x, \dot{x}, \ldots, x^{(n)}$ in  a nontrivial way, one should define $n-1$ new variables
       \be
       \lb{igreki}
       y_1 = \dot{x}, \ \ y_2 = \dot{y}_1 = \ddot{x}, \ \ \ldots, \ \  y_{n-1} = \dot{y}_{n-2} = x^{(n-1)} 
        \ee
        and the canonical momenta
        \be
        \lb{mom-def-n}
        p_{n-1}  &=& \frac {\partial L}{\partial \dot{y}_{n-1}}  \nn
        p_{n-2}  &=& \frac {\partial L}{\partial \dot{y}_{n-2}} - \dot{p}_{n-1} \nn
        \ldots &=& \ldots \nn
         p_0  &=& \frac {\partial L}{\partial \dot{x}} - \dot{p}_1
         \ee
         The Hamiltonian involves then $n$ pairs of canonically conjugated variables
         $(p_0, x), \ldots, (p_{n-1}, y_{n-1})$ and reads
         \be
         \lb{Ostr-Ham-n}
         H \ =\ p_{n-1} \dot{y}_{n-1} + \cdots  + p_0 \dot{x} - L \, .
          \ee
          Again, one can be convinced that the Hamilton equations of motion following from \p{Ostr-Ham-n}
          are equivalent to the Lagrange equations of motion
            \be
  \lb{Lagr-eqmot-n}
  \left( \frac {d}{dt} \right)^n  \frac {\partial L}{\partial x^{(n)}}  - 
  \left( \frac {d}{dt} \right)^{n-1}  \frac {\partial L}{\partial x^{(n-1)}} + \cdots
  + (-1)^n \frac {\partial L}{\partial x} \ = \ 0 \, ,
   \ee
   bearing in mind the definitions \p{igreki} and \p{mom-def-n}.
  
  \vspace{1mm}

 Let us now prove an important theorem \cite{Woodard}
 \begin{thm}
 If $n \geq 2$ and the canonical momentum $p_{n-1}$ does not vanish, the Hamiltonian function \p{Ostr-Ham-n} may acquire an arbitrary real value.
 \end{thm}
 \begin{proof}
 Rewrite \p{Ostr-Ham-n} as
 \be
 \lb{Ostr-Ham-n1}
  H \ =\ p_{n-1} \dot{y}_{n-1} + p_{n-2} y_{n-1} + \cdots  + p_0 y_1 - L( \dot{y}_{n-1} , y_1,
  \ldots, y_{n-1}, x)
    \ee
    The generalized velocity $ \dot{y}_{n-1} $ can be expressed via the variables $p_{n-1}, y_1, \ldots, y_{n-1}, x$ by solving the first equation in \p{mom-def-n} for $ \dot{y}_{n-1} $. It does {\it not} depend on other generalized momenta $p_0, \ldots, p_{n-2}$.
    Thus, the only dependence of the Ostrogradsky Hamiltonian \p{Ostr-Ham-n1}
    on these momenta is due to the contribution
     \be
    X \ = \ p_{n-2} y_{n-1} + \cdots  + p_0 y_1 \, .
      \ee
      This contribution is linear in momenta and can obviously acquire any real value. And the same is true for \p{Ostr-Ham-n1}.
      \end{proof}
      
     \begin{rem} The requirement $n \geq 2$ is essential. Ordinary dynamical systems with $n =1$ include only one generalized momentum $p_0 \equiv p_{n-1}$, 
     and the dependence of the Hamiltonian on $p_0$ need not be linear. For example, the familiar Hamiltonian $H = p^2/2$ for a free particle is quadratic in $p$
      and takes only positive values.
      \end{rem}
      
      \begin{rem} All this can be easily generalized on the case of several variables $x_j$. The theorem above holds if the dependence of the Lagrangian on at least one of the higher derivatives of some such variable is non-trivial so that the corresponding partial derivative of $L$ does not vanish.
      \end{rem}
      
      \section{Pais-Uhlenbeck oscillator}
 \setcounter{equation}0
     
      The first study of quantum systems involving higher derivatives was performed in 1950 \cite{PU}, exactly hundred years after the publication  of Ostrogradsky's paper. Pais and Uhlenbeck considered the Lagrangian
      \be
      \lb{LPU}
      L \ =\ \frac 12 \left[ \ddot{x}^2 - (\omega_1^2 + \omega_2^2) \dot{x}^2 + \omega_1^2 \omega_2^2 x^2 \right] 
        \ee
        and showed that the spectrum of the corresponding Hamiltonian is not bounded either from below or from above. 
        \footnote{The same is true for still higher derivative oscillators considered in \cite{PU}. We will not distract our attention for such systems: their physics is essentially the same as the physics of the system \p{LPU}.}
        This result looks quite natural, bearing in mind that the {\it classical} energy of such Hamiltonian is unbounded, as follows from the theorem proven in the previous section. A quantum counterpart of this theorem asserting the absence of a ground state for any nondegenerate HD system will be proven at the end of this section.
        
        The Ostrogradsky Hamiltonian corresponding to the system \p{LPU} reads
        \be
        \lb{HPU}
        H \ =\ p_x v + \frac {p_v^2}2 + \frac {(\omega_1^2 + \omega_2^2) v^2}2 - \frac {\omega_1^2 \omega_2^2 x^2}2
          \ee
          One of the Hamilton equations of motion gives $\partial H/ \partial p_x = v = \dot{x}$.
          Excluding the momenta $p_x, p_v$ from three other equations, we reproduce the Lagrange equation
          \be
          \lb{PU-eqmot}
          \left[ \frac {d^4}{dt^4} + (\omega_1^2 + \omega_2^2) \frac {d^2}{dt^2} + \omega_1^2 \omega_2^2 \right] x \ =\ 0 \, .
           \ee 
           The quantum Hamiltonian also has the form  \p{HPU} with 
           $\hat{p}_x = -i \partial/\partial x, \  \hat{p}_v = -i \partial/\partial v $. To determine
          its spectrum, it is convenient to perform the following quantum canonical transformation
          \footnote{For a general theory of quantum canonical transformations see \cite{Anderson} and references therein. One can also first perform the canonical transformation at the classical level and quantize afterwards.} 
          \cite{Man-Dav}:
      \be
      \lb{XP-canon}
    X_1 &=& \frac 1{\omega_1} \frac {\hat{p}_x + \omega_1^2 v}{\sqrt{\omega_1^2 - \omega_2^2}},
     \ \ \ \        \hat{P}_1 \equiv  -i \frac \partial {\partial X_1} \ = \  \omega_1  \frac {\hat{p}_v + \omega_2^2 x}{\sqrt{\omega_1^2 - \omega_2^2}}, \nn
    X_2 &=&  \frac {\hat{p}_v + \omega_1^2 x}{\sqrt{\omega_1^2 - \omega_2^2}},
     \ \ \ \   \ \ \      \hat{P}_2 \equiv  -i \frac \partial {\partial X_2}  \  = \ \frac {\hat{p}_x + \omega_2^2 v}{\sqrt{\omega_1^2 - \omega_2^2}} \, .
       \ee
(We assumed that $\omega_1 > \omega_2$. The case of equal frequencies will be treated a little later.)
In these terms, the Hamiltonian reduces to a difference of two oscillator
 Hamiltonians:
  \be
  \lb{H-raznost}
  H = \ \frac {\hat{P}_1^2 + \omega_1^2 X_1^2}2  - \frac{\hat{P}_2^2 + \omega_2^2 X_2^2}2 \, .    \ee
    Its spectrum is 
    \be
    \lb{spec-PU}
    E_{nm} \ =\ \left(n + \frac 12 \right) \omega_1 - \left(m + \frac 12 \right) \omega_2 
     \ee
     with positive integer $n,m$.
     
     One may ask: how come the spectrum is asymmetric under interchange $\omega_1 \leftrightarrow \omega_2$, whereas the Lagrangian \p{LPU} and the Hamiltonian \p{HPU} have this symmetry?
     To understand this note that the canonical transformations \p{XP-canon} do not have this symmetry, they involve $\sqrt{\omega_1^2 - \omega_2^2}$ rather than   $\sqrt{\omega_2^2 - \omega_1^2}$. If we assumed $\omega_2 > \omega_1$, the appropriate canonical transformation would be modified and the sign of energy in \p{spec-PU} would be reversed. 
     
     Consider the case $\omega_1 = 2 \omega_2$. Then
    \be
    \lb{spec-otn-2}
    E_{nm} \ =\ \left(2n -m + \frac 12 \right) \omega_2  \, .
     \ee  
     We see that the spectrum is discrete, not bounded and that each level is infinitely degenerate.
     The same properties 
     hold for any rational ratio $\omega_1/\omega_2$.
     
     An interesting situation arises for incommensurable frequencies. In that case, the spectrum is
     {\it everywhere dense}: for any energy $E$ one can find an arbitrary close eigenvalue $E_{nm}$.
     However, the spectrum is not continuous, but rather {\it pure point}: the wave functions of all
     eigenstates are normalizable.

In the degenerate case, $\omega_1 = \omega_2 \equiv \omega$, the situation is somewhat more complicated. 
We cannot perform the canonical transformation
\p{XP-canon} anymore and the Hamiltonian is not reduced to a difference of two oscillator Hamiltonians. The solution of the problem
can be obtained by the following trick \cite{PU,Bolonek}.

At the first step, we get rid of the potential terms $\propto x^2$ and $\propto v^2$ in the Hamiltonian 
\be
\lb{HPU-degen}
 \hat{H} \ =\ \hat{p}_x v + \frac {\hat{p}_v^2}2 + \omega^2 v^2  - \frac {\omega^4  x^2}2
  \ee
  by representing its eigenfunctions as 
  \be
\lb{Psi-phi}
\Psi(x,v) \ =\ e^{-i\omega^2 xv} \phi(x,v) \, .
  \ee
The Hamiltonian acting on $\phi(x,v)$ has the form
 \be
\lb{H-phi}
\hat{H}_\phi \ =\ \frac 12 \hat{p}^2_v + v \hat{p}_x - \omega^2 x \hat{p}_v \, .
  \ee
At the second step, we perform the quantum canonical transformation

\be
\lb{Bolonek-canon}
 x \to  \frac 1\omega x
  + 
\frac 1{4\omega^2} \hat{p}_v, \ \ \  \hat{p}_x \to \omega \hat{p}_x, \ \ \  v \to v + \frac 1{4\omega} \hat{p}_x, \ \ \ 
\hat{p}_v \to \hat{p}_v \, .
 \ee
The transformation \p{Bolonek-canon} is the superposition of the scale transformation  
$x \to x/\omega, \hat{p}_x \to \omega \hat{p}_x$
and the unitary transformation 
$$\hat{O} \to \exp\left\{\frac {i\hat{p}_x \hat{p}_v}{4\omega} \right\} \hat{O} \exp\left\{-i\frac {\hat{p}_x \hat{p}_v}{4\omega} \right\} \, . $$
 It gives the new Hamiltonian 
\be
\lb{H-varphi}
\hat{H}'_\phi \ =\ \frac {\hat{p}_x^2 + \hat{p}^2_v}4 + \omega (v \hat{p}_x - x \hat{p}_v )
  \ee
acting on the wave function $\phi'(x, v)$  related to $\phi(x, v)$ according to
 \be
\phi(x,v) \ =\ \exp \left\{ \frac i{4\omega^2} \frac {\partial^2}{\partial x \partial v } \right \}  \phi'(\omega x, v) \, .
  \ee

The first term in the Hamiltonian \p{H-varphi} describes free motion in the transformed $(x,v)$ plane, and the second term is proportional
to the ``angular momentum'' operator 
 \be
\lb{l}
\hat{l} = v \hat{p}_x - x \hat{p}_v \, .
  \ee
 Expressed in the original variables, this operator reads
  \be
\lb{l-orig}
\hat{l} \ =\ \frac {v \hat{p}_x}{2\omega} - \frac \omega 2 x \hat{p}_v + \frac 1{4\omega} \left( \hat{p}_v^2 - 
\frac {\hat{p}_x^2}{\omega^2} \right) + \frac {3\omega }4(v^2 - \omega^2 x^2) \, . 
  \ee
A dedicated reader may check that it commutes with the original Hamiltonian \p{HPU-degen}.

The eigenfunctions of \p{H-varphi} are the same as for the free Laplacian, but the energy is shifted by $l\omega$, where
$l$ is the integer eigenvalue of \p{l}. They are
  \be
\lb{varphi-state}
\phi'_{lk}(\omega x,v; t)\ \propto \  J_l(kr) e^{-il\theta} e^{-it (l\omega + k^2/4)} \, ,
  \ee
where $(r, \theta)$ are the polar coordinates in the plane $(\omega x, v)$. 

The eigenfunctions of the original Hamiltonian are
 \be
\lb{Psi-state}
\Psi_{lk}(x,v; t) \ \propto \  e^{-i\omega^2 xv} 
 \exp \left\{ \frac i{4\omega^2} \frac {\partial^2}{\partial x \partial v} \right\} 
\left[ J_l\left( k \sqrt{v^2 + \omega^2 x^2} \right) \left( \frac {\omega x - iv}{\omega x + iv} \right)^{l/2}   \right]
e^{-it (l\omega + k^2/4)} 
 \ee
 They are not normalizable and describe a continuum spectrum rather than a pure point spectrum characteristic for the nondegenerate system with $\omega_1 \neq \omega_2$.  Any real energy is admissible, and each energy level is infinitely degenerate:
the eigenfunctions $\Psi_{l+1,k}$ and $\Psi_{l, \sqrt{k^2 + 4\omega}}$ have the same energy.

The difference in the quantum behaviour of the degenerate ($\omega_1 = \omega_2$) and non-degenerate ($\omega_1 \neq \omega_2$) systems matches the difference in their classical behaviour. The motion described by the Hamiltonian \p{H-raznost} is finite and that is why the spectrum is pure point. And in the degenerate case the situation is different. The equation of motion
           \be
          \lb{PU-eqmot-degen}
          \left[ \frac {d^4}{dt^4} + 2\omega^2  \frac {d^2}{dt^2} + \omega^4 \right] x \ =\ 0 \, .
           \ee 
admits not only ordinary oscillatory solutions $x(t) \propto e^{i\omega t}$, but also the solutions
\be
\lb{lin-rost}
 x(t) \propto t e^{i\omega t}\, , 
  \ee
where the amplitude of oscillations grows with time. An infinite classical motion produces a
continuum spectrum in the quantum problem.

One can observe that both in the non-degenerate and in the degenerate case both the classical and the quantum dynamics of the system are quite benign. The wave functions are known explicitly.
The Hamiltonian is Hermitian and the quantum evolution is unitary. For example, the evolution
operator for the system \p{H-raznost} is simply a product of the evolution operators for the individual oscillators:
  \be
  \lb{KK} 
{\cal K}(X_1', X_2'; X_1, X_2; t) \ =\ \sum_{n=0}^\infty \psi^*_n(X_1') \psi_n(X_1) e^{-i\omega_1 t(n+1/2)} \ 
\sum_{m=0}^\infty \psi^*_m(X_2') \psi_m(X_2) e^{i\omega_2 t(m+1/2)} \, ,
  \ee
  where $\psi_n(X)$ are the standard oscillator eigenfunctions.

  One can meet in the literature a statement that the Pais-Uhlenbeck oscillator with degenerate frequencies is not unitary
  due to the presence of Jordan blocks \cite{BM-nonunit}. But this is not correct. Infinite-dimensional Jordan blocks
  appear in this problem if one tries to describe the motion in terms of the oscillator wave functions, as if the motion were
  finite. And their appearance is simply a manifestation of the fact that the spectrum in this case is continuous, and not of the violation
  of unitarity. We refer the reader to Appendix A and also to Refs. \cite{PU-Sigma},\cite{except}, where we clarify this point.
  
  The evolution operator \p{KK} can also be evaluated by calculating numerically the {\it Minkowskian} path integral, using e.g. the methods of \cite{Blin}. This integral can be expressed either in the Lagrangian form
     \be
   \lb{L-path}
   \sim \int \prod_t dx(t) \exp \left\{  i \int dt\, L(\ddot{x}, \dot{x}, x) \right\}
    \ee
    or in the Hamiltonian form
     \be
   \lb{H-path}
   \sim \int \prod_t dx(t) dv(t)  dp_x(t)  dp_v(t)    \exp \left\{  i \int dt\, [p_v \dot{v} + p_x \dot{x} - H(p_v, p_x; v, x)] \right\} \, .
    \ee
    Indeed, substituting in \p{H-path} the Hamiltonian \p{HPU} and integrating over 
    $\prod_t dp_x(t)$, we obtain the factor
      \be
       \lb{factor}
       \prod_t \delta [v(t) - \dot{x}(t) ] \, .
        \ee
        Performing the integral over $\prod_t dv(t)$ with this factor and doing the Gaussian integral over $\prod_t dp_v(t)$, we
        reproduce \p{L-path} with the Lagrangian \p{LPU}. 
        
        In ordinary quantum theories including a well-defined vacuum state we can perform the Wick rotation
        $ t \to -i\tau$ and define an Euclidean path integral. This does not work, however, for the integral \p{H-path}. 
        Already the integral 
          \be
          \lb{int-px}
          \prod_\tau    \int_{-\infty}^\infty  dp_x(\tau)    
            \exp \left\{   \int d\tau \, p_x(\tau) \left[ i \frac {dx(\tau)}{d\tau} - v(\tau) \right]  \right \}
           \ee
        diverges and does not give anything like \p{factor}. (The latter circumstance means that the variable $v(\tau)$ can no longer be interpreted as the velocity.) 
      The fact that the Euclidean path integral diverges agrees well with that the evolution operator \p{KK} is not defined at imaginary time --- at $t = -i\tau$ the second series diverges with a vengeance.   
        
        Another way of going into Euclidean space was suggested in \cite{Hawking}. Instead of performing an analytic continuation of the  Euclidean path integral \p{H-path}, one could try to do it at the Lagrangian level, capitalizing on the fact that the Euclidean
        counterpart of the Lagrangian \p{LPU} is positive definite. Trading $t$ for $i\tau$ 
        \footnote{The standard Wick rotation $t \to -i\tau$ works if one adds the extra minus sign in the definition of $L$.}
        in \p{L-path}, one obtains a converging path integral.
        
        However, if one tries to take the expression for the Euclidean evolution operator thus derived and goes back to real Minkowsky time, one obtains something which does not describe a unitary evolution with a Hermitian Hamiltonian \cite{Hawking,Maslanka}.
        It is thus better not to do so, but  keep the time real all the time.

  \vspace{1mm}
  
   We analysed in detail the simple system with the Hamiltonian \p{HPU} and showed that its spectrum  has no ground state. But the absence of the ground state is a characteristic feature of all higher-derivative theories, not only of the Pais-Uhlenbeck oscillator. One can prove the following theorem:
   
   \begin{thm}
   The quantum conterpart of the Ostrogradsky  Hamiltonian for a non-degenerate higher-derivative system has no ground state.
   \end{thm}
   \begin{proof}
   We will prove it for the simplest system \p{Lagrx2} (with $\partial L/\partial \ddot{x} 
   \neq 0$).  \footnote{A similar proof was given in Ref. \cite{Raidal}.} A generalization for the Lagrangians involving still higher derivatives and many degrees of freedom is obvious. Also we will assume that the spectrum of the Hamiltonian is discrete. Regarding the systems with continuous spectrum, one should first regularize them by ``putting them in a box'':   the range where the variables $x,v$ may change should be made finite, which  makes the motion finite and the spectrum discrete. Then we can apply our reasoning and then remove the regularization, sending the size of the box to infinity.
   
   Consider the quantum Hamiltonian \footnote{The ordering ambiguity is irrelevant for the arguments below.}
   \be
   \lb{H-Ostro}
  \hat{H} \ =\  a(\hat{p}_v, x, v)\hat{p}_v +  v\hat{p}_x - L \left[a (\hat{p}_v, x, v), v, x \right] \, .
    \ee
    Its middle term is
    $-iv \partial/ \partial x$. Suppose that this Hamiltonian has a normalized ground state 
    $\Psi_0(x, v)$ of energy $E_0$.  We take a variational ansatz
    \be
    \lb{Psi-c}
    \Psi_c(x,v) \ =\ e^{icx} \Psi_0(x,v) \, .
     \ee 
     Obviously, 
     $$ \int \Psi_c^* \Psi_c \, dx dv \ =\ \int \Psi^*_0 \Psi_0 \, dx dv \ =\ 1 \, . $$
     The function $\Psi_c(x,v)$ is not an eigenfunction of the Hamiltonian:
      $$\hat{H} \Psi_c \ =\ E_0 \Psi_c + cv \Psi_c  \neq \lambda \Psi_c \, .$$
      This applies also to any linear combination of $\Psi_c$,
      $$ \Psi_f \ =\ \int f(c) \, \Psi_c \, dc  = {\tilde f}(x) \Psi_0(x,v) $$
      [${\tilde f}(x) $ is the Fourier transform of $f(c)$]:
      $$\hat{H} \Psi_f \ =\ E_0 \Psi_f -iv \frac {\partial  {\tilde f}(x)} {\partial x} 
      \Psi_0 \neq \lambda \Psi_f \, .$$
     Then $\Psi_c$ and $\Psi_f$ have nonzero projections on excited states and  the variational energy
     \be
     \lb{var-energy}
     E(c) \ =\ \int \Psi_c^* \hat{H} \Psi_c \, dx dv 
      \ee
      should exceed $E_0$. However, a direct calculation gives
      \be
      \lb{E(c)}
       E(c) \ =\ E_0 + c \int v |\Psi_0|^2 \, dx dv \, .
        \ee
      If the integral $\int v |\Psi_0|^2 \, dx dv $ is different from zero, we can shift the energy down by choosing an appropriate sign of $c$. If this integral is zero, the energy does not depend on $c$ and does not grow, as it should  if $\Psi_0$ would be a ground state. Thus, the original assumption of the existence of the ground state was wrong.

   \end{proof}
    
    \begin{rem}
    Applying the same reasoning, but reversing the sign of $c$ in \p{E(c)}, we may conclude that a ``sky state'', a state with the maximal energy, is also absent --- the spectrum is not bounded neither from below, nor from above. 
     \end{rem}
     
     \begin{rem}
     A meticulous mathematician would notice that we have proven the absence of the ground state, but this alone does not imply yet that
     the energy levels of the Hamiltonian go all the way down to negative infinity. There is a logical possibility that the levels represent an infinite sequence with $E_{n+1} < E_n$, but the spacing between the adjacent levels goes down for large $n$ so that the sequence has a finite limit \, $\lim_{n \to \infty} E_n = E^* > -\infty$. 
     
     But this is really an {\it exotic} possibility --- I am not aware of quantum systems exhibiting such behaviour. Probably, a clever expert in functional analysis could exclude it...
     \end{rem}

  \section{Philological digression}
  \setcounter{equation}0
  
  We have proven that the nondegenerate quantum higher-derivative systems have no ground state. How should one {\it call} this phenomenon?
  This question is actually not so irrelevant. Experience shows that an inexact, not carefully chosen name may excite false associations and lead   eventually to wrong claims.
  
  Traditionally, one says that such bottomless systems involve {\it ghosts}. Let us explain where this name came from and what it means. By ``ghosts'' one usually has in mind the states with negative norm. Negative norm means negative probability and production of such states means violation of unitarity. Such ghosts appear e.g. in gauge theories. These are scalar photons in QED, and  scalar gluons and Faddeev-Popov ghosts in non-Abelian theories. 
  
   In the framework of the Gupta-Bleuler quantization procedure for QED, one introduces the creation and annihilation operators
   $a_\mu^\dagger (\vecg{k})$ and $a_\mu(\vecg{k})$ for all four photon polarizations, which satisfy the following
   commutation  relations:
     \be
     \lb{Gupta}
     [a_\mu(\vecg{k}), a^\dagger_\nu(\vecg{q})] \ =\ - \eta_{\mu\nu} \delta (\vecg{k} - \vecg{q}) \, .
       \ee
       The commutator $[a_0, a^\dagger_0]$ is then negative. That means  that the scalar photon state has a negative norm.
       Indeed, introduce the Fock vacuum $|\Phi \rangle$ whose norm is positive. We then have
       \be
       \lb{neg-norm}
       \left| |a_0^\dagger |\Phi \rangle \right|^2 \ =\ \langle\Phi|a_0 a_0^\dagger | \Phi \rangle \ =\ \langle \Phi| [a_0, a_0^\dagger] | \Phi \rangle \ =\ -\langle \Phi|\Phi \rangle  \, .
        \ee
        In gauge theories, such ghosts are harmless: one can define a reduced physical space involving only
        the physical states with positive norm (transverse photons or transverse gluons) and show that the ghosts states are not created in collisions of physical particles. \footnote{This is explained in many textbooks including my own book \cite{kniga}.}
        
        Let us go back to the Pais-Uhlenbeck oscillator. Consider the non-degenerate case, which is more simple. Introducing the creation and annihilation operators in the usual way, we may write the quantum Hamiltonian as 
         \be
         \lb{HPU-via-a}
         H \ =\ \omega_1 a_1 a_1^\dagger - \omega_2 a_2 a_2^\dagger + C \, .
         \ee
         The excitations of the second term give a tower of states with decreasing energies all the way down to $-\infty$. Suppose, however, that we do not like these negative energies so much that we want to get rid of them by any cost. 
         Then one can formally define a state $|\Phi \rangle$  which is annihilated by $a_2^\dagger$ rather than by $a_2$. 
         This state is strange and unhandy: its wave function is \cite{Ilhan}
          \be
          \lb{wrong-vac}
           \Phi(X_2)   \ \propto \ \exp \left\{ \frac {\omega_2}2 X_2^2 \right\} 
             \ee
         so that the state is not normalizable. Still formally the state $|\Phi \rangle $ has the lowest energy. We can even bring it to zero or to any other finite value by choosing an infinite positive constant $C$ in \p{HPU-via-a}. Then the operator $a_2$ acting on $|\Phi \rangle$ would increase its energy, and we can {\it call} it a creation operator $b^\dagger$. And
         $a_2^\dagger$ is now interpreted as the annihilation operator $b$. The price one has to pay for that  is that the commutator
         $[b, b^\dagger]$ is now negative and the states describing excitations above the ``vacuum'' $|\Phi \rangle$ are ghost states, they have negative norm. 
         
         Personally, I find this construction rather awkward. It is much better to talk about negative energies, keeping the norm positive. But,  historically, people thought in these terms and that is why they  (wrongly) believed for a long
         time that higher derivatives necessarily entail the violation of unitarity. 
         
R. Woodard suggested to call this phenomenon ``Ostrogradskian instability'' without the reference to ghosts.
This name has, however, its own drawbacks. First of all, Ostrogradsky knew nothing  about quantum Hamiltonians and their spectra. The observation of the absence of the ground states in HD quantum systems belongs not to him, but to Pais and Uhlenbeck. \footnote{And the remark that the classical energy of a generic HD system is unbounded, the Theorem 1 of Sect. 2, which  Woodard calls ``The theorem of Ostrogradsky'',  belongs not to Ostrogradsky, who developped the general classical Hamiltonian formalism but did not study the dynamics of particular HD systems, but to Woodard himself.}

In addition, the word ``instability'' invokes wrong images. Having heard it, a physicist imagines a ball on the top of the hill and thinks of the exponential growth of deviations from the equilibrium. But for certain HD systems, there is no such growth. The spectrum has no bottom, but there is no instability --- neither at the classical nor at the quantum level.

Thus, we prefer not  to use this word and not invent anything new, but rather to stick to the traditional and more familiar to most people  name ``ghost'', not invoking, however, the negative metric description. For us, a
 {\it ghost system} is by definition a system where the Hamiltonian does not have a ground state but involves {\it ghosts} ---  oscillatory-type excitations with negative energies, with all the states of this Hamiltonian having positive norm.
  We will call such exscitations ghosts, but will distinguish the systems with {\it benign} ghosts, where the quantum problem is well defined, as it is for the Pais-Uhlenbeck oscillator, and the systems with {\it malignant} ghosts involving
  a collapse and loss of unitarity. The latter systems are more common. They will be discussed in Sect. 5. And some examples 
  of the benign ghost systems will be given in Sects. 6,7. 
  
  I also have to comment here on another a little confusive issue. C. Bender and P. Mannheim suggested to bust 
  ghosts by replacing the original Hilbert space spanned over the eigenfunctions $\Psi(x,v)$ of the Hamiltonian
  \p{HPU} or, which is equivalent for different frequencies, the Hilbert space spanned by the oscillator wave functions of the Hamiltonian \p{H-raznost} by {\it another} Hilbert space involving the functions depending on real $X_1$ and {\it imaginary} $X_2 = iY_2$ and normalized in that region \cite{BM}. Then the wave function
  \p{wrong-vac} becomes a good normalizable vacuum state, $\int_{-\infty}^\infty |\Phi|^2 dY_2 < \infty$, and all other eigenstates have positive energies (and positive norm).
  
  Personally, I do not quite see a point of doing so. One should understand that this complexification brings
  us to  {\it another} quantum problem having little to do with the original one. No ghosts in the former does not mean no ghosts in the latter.

  \section{Including interactions}
  \setcounter{equation}0

  We discussed so far only mechanical systems, but it is easy to write down a field-theory generalization of \p{LPU}. We can write
    \be
    \lb{LPU-field}
    {\cal L} \ =\ \frac 12 \left[ (\Box \phi)^2 - (M_1^2 + M_2^2) (\partial_\mu \phi)^2 + M_1^2 M_2^2 \phi^2 \right] \, .
     \ee
     When $M_1 \neq M_2$, one can go over to the Hamiltonian, perform an approriate canonical transformation and bring it again in the Lagrangian form to obtain
      \be
    \lb{L-raznost-field}    {\cal L} \ =\ \frac 12 \left[ (\Box \Phi_1)^2 - M_1^2 \Phi_1^2 \right]  - 
    \frac 12 \left[ (\Box \Phi_2)^2 - M_2^2 \Phi_2^2 \right]  \, .
     \ee
     This Lagrangian includes an ordinary scalar field $\Phi_1$ and a ghost field (i.e., in our terminology, a field whose quanta carry negative energies) 
     $\Phi_2$. We can expand each field in Fourier modes, in which case the Lagrangian is split into an infinite number of noninteracting sectors 
     with a definitite momentum $\vecg{p}$. In each such sector, the dynamics is described by the Hamiltonian \p{H-raznost} with
      $\omega_{1,2}^2 = \vecg{p}^2 + M_{1,2}^2$.
      In free theory, the negative sign of the second term in \p{H-raznost} or \p{L-raznost-field} is irrelevant. The classical equations of  motion and their solutions do not depend
      on this sign. The spectra of the quantum Hamiltonians are different for different signs, but in the absence of interactions, energy does not really play a dynamical role and serves only for bookkeeping.
      
      A trouble may set it, however, if interactions are included. Heuristically, one could expect in this case an instability due to copious production of ghosts. Let us see whether an interactive system has or not this instability.
      
      We modify the Lagrangian \p{LPU} by adding there  nonlinear terms. 
      Consider for example the system \cite{malicious}
 \be
      \lb{LPU-nonlin}
      L \ =\ \frac 12 \left[ \ddot{x}^2 - 2\omega^2  \dot{x}^2 + \omega^4  x^2 \right] - \frac 14 \alpha
      x^4 \, .  
        \ee
The classical equations of motion read 
  \be
  \lb{eqmot-PUnonlin}
   \left( \frac {d^2}{dt^2} + \omega^2 \right)^2  x - \alpha x^3 \ =\ 0 \, .
    \ee 
The classical trajectories depend on four initial conditions. There is an obvious stationary point
 \be
 \lb{stat-point}
 x(0) \ = \ \dot{x}(0) \ = \ \ddot{x}(0) \ = \ x^{(3)}(0)  \ =\ 0\, .
  \ee
  The behaviour of the system at the vicinity of this point depend on the sign of $\alpha$. If $\alpha <0$, the Ostrogradsky Hamiltonian  acquires an extra negative contribution to the energy and all the trajectories other than $x(t) = 0$ are unstable --- they run away to infinity in finite time. This instability is even worse that Lyapunov's exponential growth of perturbations. We are dealing here with a {\it collapse}: the classical dynamical problem is not well posed beyond a certain time  horizon.
 
  The situation is better  for positive $\alpha$. The stationary point \p{stat-point} lies in the center
  of an ``island of stability'' --- the trajectories with initial conditions at its vicinity do not go astray, but 
  exhibit an oscillatory behavour. However, this island has a shore. When the deviations of intial conditions from \p{stat-point} are large enough, a trajectory collapses. To chart this shore, one should perform a numerical study. Assume for simplicity that $\dot{x}(0), \ddot{x}(0)$ and $ x^{(3)}(0)$ stay zero, and $x(0) = c > 0$. Then a critical value $c_{\rm crit}$ can be found, above which the trajectories go astray and run to infitiny, but  the system only exhibits benign oscillations around zero
 when $c < c_{\rm crit}$.  This value is roughly  $c_{\rm crit} \approx 0.3 \, \omega^2/\sqrt{\alpha}$ (the dependence on $\omega$ and $\alpha$ follows, of course, from simple scaling arguments).
 For illustration, we plotted in Fig.\ref{island} the solution to the equations of motion \p{eqmot-PUnonlin} for $\omega = \alpha = 1$ and $x(0)$
 just above $ c_{\rm crit}$. After some quasiharmonic oscillations, the trajectory finally goes astray and runs to infinity. And for $x(0) < c_{\rm crit}$, it
 keeps oscillating forever.

 \begin{figure}[ht!]
   \begin{center}
 \includegraphics[width=5 in]{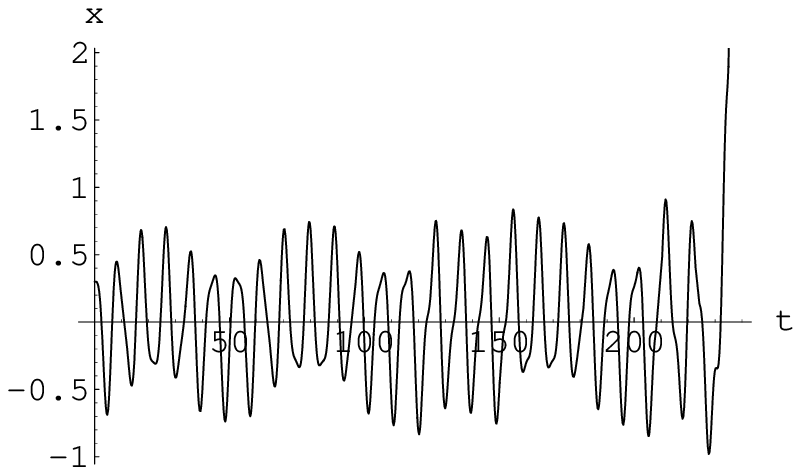}
    \end{center}
\caption{Oscillating and collapsing}
\label{island}
\end{figure} 
The same behaviour is seen if three other initial conditions are shifted from zero. The island has a finite area (or rather a finite phase space volume).

A similar island of stability was  observed in \cite{Carroll} for the model Hamiltonian
\footnote{The authors of this paper were interested in the toy Hamiltonian \p{HCarrol} in association with a certain cosmological problem.} 
\be
\lb{HCarrol}
H \ =\ \frac 12 \left[ p_x^2 - (p_y^2 + y^2) + x^2y^2 \right] \, ,
 \ee
 where a free particle is coupled to the ghost oscillator. If we pose the initial conditions
  \be
  \lb{incond}
  x(0) = y(0) = a, \ \ \ \ \ \ \ \ \ \ p_x(0) = p_y(0) = 0 \, ,
    \ee
 the system undergoes benign oscillations at the vicinity of the origin if 
 $ a \la 0.59$ and goes astray, running to the infinity, for larger deviations.
  
 What happens in the quantum case? One cannot give a quite definite answer to this question without a special numerical study, but certain heuristic arguments that the quantum problem is also malignant and collapsing in this case will be given below.
 
 Classical and quantum collapse are well known for certain special ordinary (not HD) systems. The simplest such system describes the 3-dimensional motion of a particle with the attractive potential 
    \be
    \lb{-kappa/r2}
    V(r) = -\kappa/r^2\, .
     \ee
      Classically, for certain initial conditions, the particle falls to the center
 in a finite time. The quantum dynamics of this system depends  on the value of $\kappa$. If $m\kappa < 1/8$ (where $m$ is the mass of the particle), the ground state
 exists and unitarity is preserved.  If $m\kappa > 1/8$, the spectrum is not bounded from below and, which is worse, the quantum problem cannot
 be well posed until the singularity at the origin is removed. For example, one can pose $V(r) = -\kappa/r^2 $ for $r >a$ and 
 $V(r) = -\kappa/a^2 $ for $r \leq a$.  The spectrum  then depends on $a$ \cite{Popov}.  Without such a cutoff, the probability ``leaks''  into the singularity and unitarity is violated. And for $m\kappa < 1/8$ quantum fluctuations cope successfully with the attractive force of the potential and prevent the system from collapsing.
 
 Going back to the system \p{LPU-nonlin}, a variational analysis  shows that, in constrast to the system \p{-kappa/r2} with small $\kappa$,  the ground state in the spectrum is absent \cite{PU-Sigma}. One can {\it conjecture} that for all such ghost-ridden systems involving collapsing classical trajectories, their ghosts are malignant and unitarity is violated by the same mechanism as for the strong attractive potential \p{-kappa/r2}. On the other hand, one can also conjecture that 
 
 {\it If the classical dynamics of the system is benign, its quantum dynamics is also benign, irrespectively of whether the spectrum has or does not have a bottom.}
 
 \section{Benign ghosts: classical and quantum mechanics}
 \setcounter{equation}0
 
 In some special cases, a nontrivial interacting ghost system may be completely benign: all the classical trajectories (and not only the trajectories restricted to a
 limited region of phase space) are stable and behave well at all times. The quantum problem is also well posed.
 
 The first example of such a system was found in \cite{Robert} \footnote{
 It was found by exploring a certain higher-derivative supersymmetric system. The details are given in Appendix B.}. The Hamiltonian involves two pairs of the dynamic variables, $(x, p)$ and $(D, P)$, and reads
   \be
   \lb{H-Robert}
   H \ =\ pP + DV'(x) \, ,
     \ee
     where $V(x)$ is an arbitrary smooth even function. The simplest nontrivial case is 
      \be
      \lb{V-Robert}
      V(x) \ =\  \frac {\omega^2 x^2}2 + \frac {\lambda x^4}4 \, , \ \ \ \ \ \ \lambda > 0
      \, .
         \ee
The kinetic part of the Hamiltonian \p{H-Robert} is not positive definite, and, obviously, the spectrum of its quantum counterpart 
has no bottom. 
If introducing
  \be
  \lb{X12-Robert}
  X_{1,2} \ =\  \sqrt{\frac \omega 2} x \pm  \frac 1 {\sqrt{2\omega}}  D \, , \ \ \ \ \ \ 
   P_{1,2} \ =\  \frac 1 {\sqrt{2\omega}} p \pm   \sqrt{ \frac \omega 2}  P
    \ee
    and the corresponding  canonical momenta, the Hamiltonian \p{H-Robert}
    acquires the form
      \be
      \lb{H-raznost-Robert}
      H \ =\ \frac{P_1^2 + \omega^2 X_1^2}2 - 
      \frac{P_2^2 + \omega^2 X_2^2}2 + \frac {\lambda}{4\omega} (X_1 - X_2) (X_1 + X_2)^3 \, .
       \ee
       In other words, this is the Hamiltonian \p{H-raznost} with degenerate frequencies, where an extra quartic interaction of a special form is added.
       
       A nice distinguishing feature of the system \p{H-Robert} is its exact solvability. Indeed, it involves besides $H$ another integral of motion:
        \be
        \lb{N}
        N \ =\ \frac {P^2}2 + V(x) \, .
          \ee
          The Poisson bracket $\{H, N\}$ vanishes. This allows one to find the solution analytically.
   
   The classical equations of motion are 
   \be
   \lb{eqmot-Robert}
   \ddot{x} + V'(x) \ =\ 0\,, \ \ \ \ \ \ \ \ \ \ \ \ddot{D} + V''(x) D = 0 \, .
    \ee
    The first equation is especially simple. It describes oscillations in the quartic potential
     \p{V-Robert}. The solutions are the elliptic functions whose parameters depend on the integral of motion $N$ :
      \be
      \lb{cn}
      x(t) \ =\ x_0 {\rm cn} [\Omega(t-t_0), k] 
       \ee
       where  
       \be
       \alpha = \frac {\omega^4}{\lambda N}, \ \ \ 
       \Omega = [\lambda N(4+\alpha)]^{1/4}, \ \ \ k^2 \equiv m  = \frac 12 \left[ 1 - \sqrt{\frac \alpha{4 + \alpha}} \right], \nn
        x_0 = \left( \frac N\lambda \right)^{1/4} 
       \sqrt{\sqrt{4+\alpha} - \sqrt{\alpha}}   \, .    
  \ee
  Here $k$ is the parameter of the Jacobi elliptic functions \cite{Ryzhik}
  \footnote{Recall that cn$(\tau)$ is a periodic function with the period
   \be
   \lb{period}
   4K \ =\ 4\int_0^{\pi/2} \frac {d\theta}{\sqrt{1 - k^2 \sin^2 \theta}} \, .
    \ee }

The equation for $D$ represents an elliptic variety of the Mathieu equation. Generically, it is not the simplest kind of equation, but in our case the solutions can be found in an explicit form. One of the solutions is 
 \be
 \lb{D1}
 D_1(t) \ \propto \dot{x}(t) \ \propto {\rm sn}[\Omega t, k] {\rm dn}[\Omega t, k] 
   \ee
   (we have chosen $t_0 = 0$).
   The second solution can be found from the condition that the time 
   derivative of the Wronskian $W = \dot{D}_1 D_2 - \dot{D}_2 D_1 $ vanishes. 
   We find 
   \be 
   \lb{D2}
   D_2(t) \propto \dot{x}(t) \int^t \frac {dt'}{\dot{x}^2(t')} \ \propto \  
   {\rm sn}[\Omega t, k] \, {\rm dn}[\Omega t, k]   \int^t \frac {dt'} {{\rm sn}^2 [\Omega t', k] \, {\rm dn}^2[\Omega t', k]} \, .
     \ee
     Two independent solutions \p{D1} and \p{D2} exhibit oscillatory behaviour with constant and linearly rising amplitude, correspondingly (see Fig. \ref{Drost}). This linear growth has the same nature as in \p{lin-rost} and does not represent a problem.
     
     \begin{figure}[h]
   \begin{center}
 \includegraphics[width=4 in]{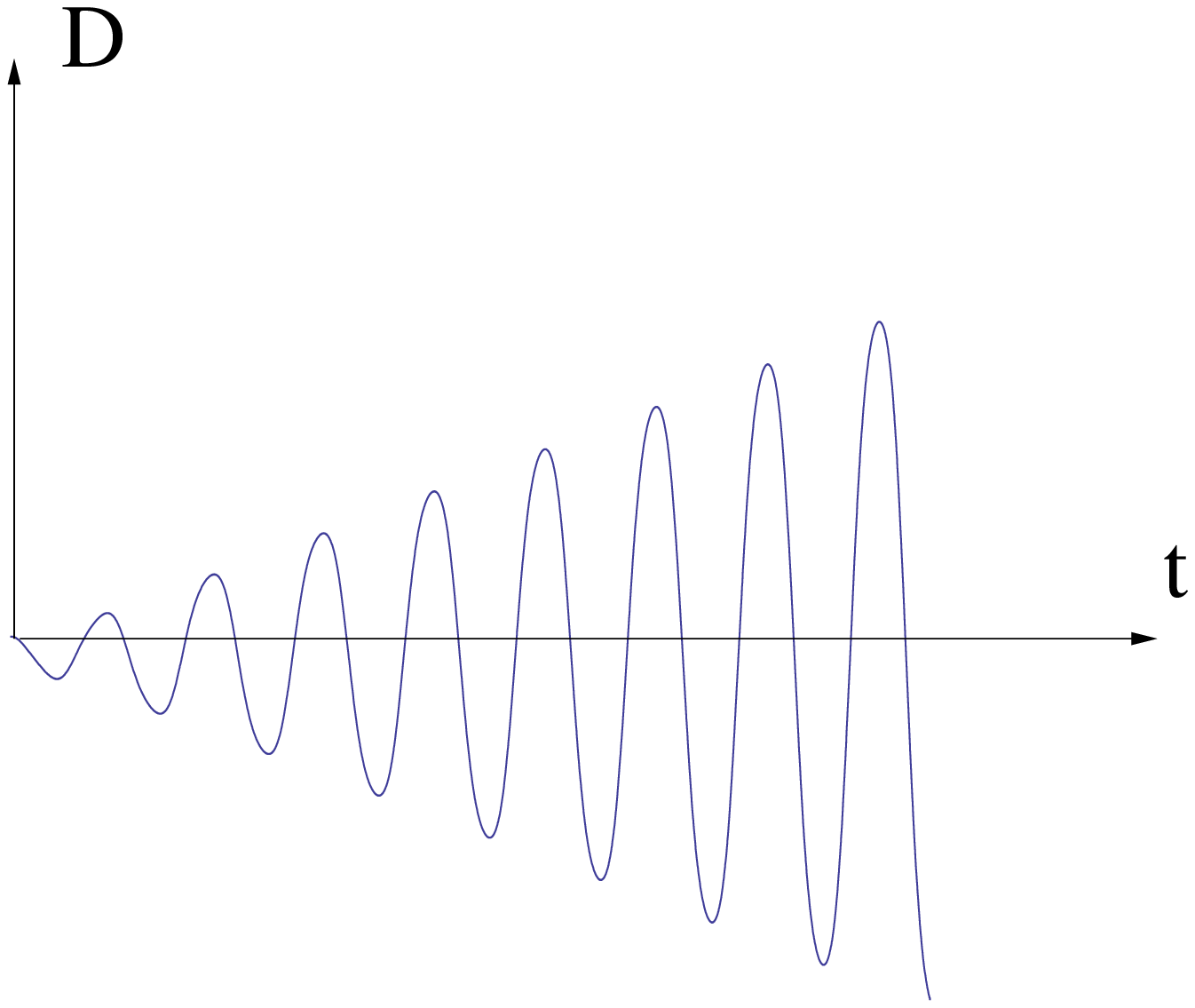}
    \end{center}
\caption{A typical behaviour of $D(t)$, as follows from the solution of \p{eqmot-Robert}.} 
\label{Drost}
\end{figure}

  \begin{figure}[h]
   \begin{center}
 \includegraphics[width=6 in]{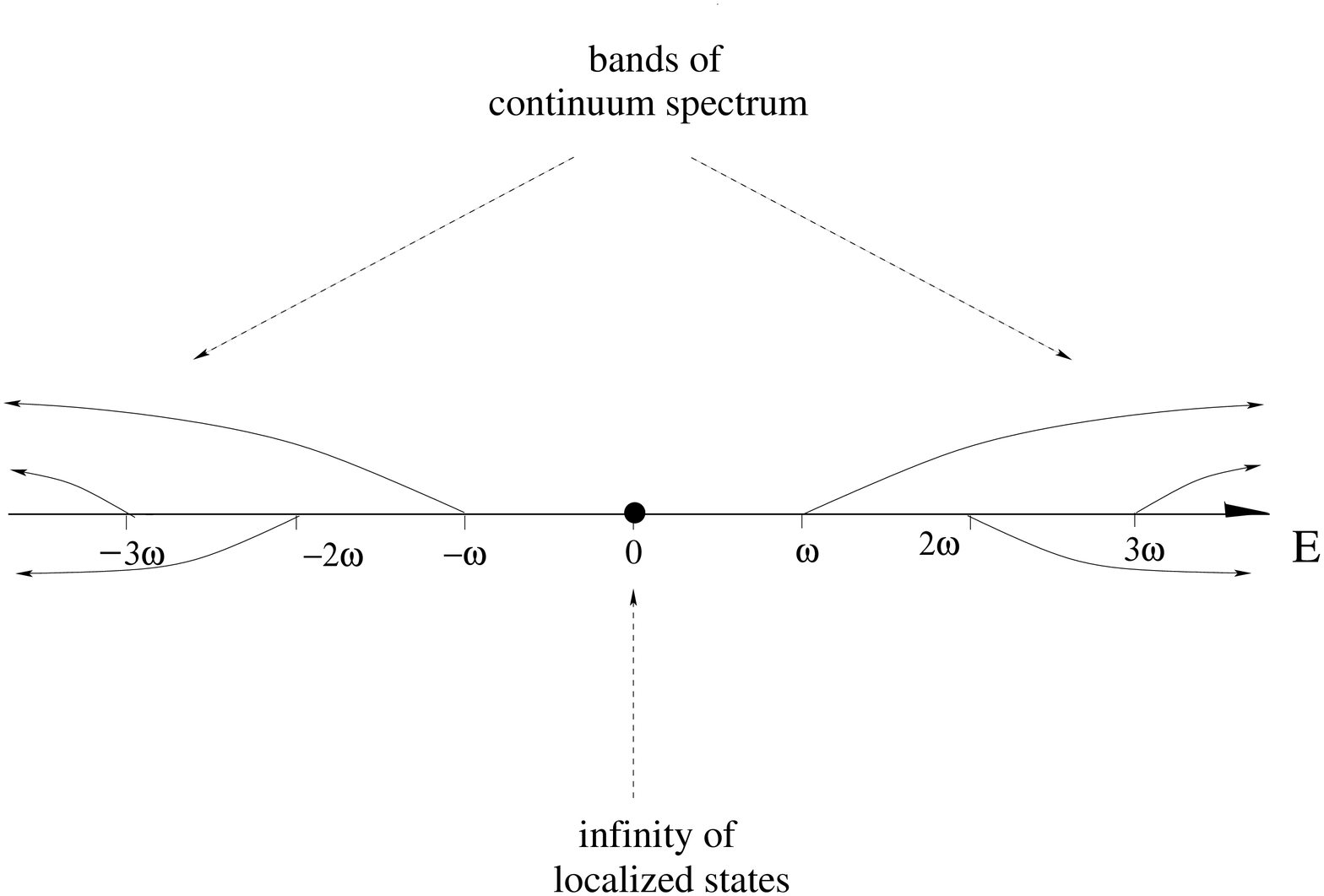}
    \end{center}
\caption{Spectrum of the Hamiltonian \p{H-Robert}.} 
\label{spectr-Robert}
\end{figure} 

The eigenvalues and eigenfunctions of the quantum Hamiltonian can also be found explicitly in this case \cite{Robert}. Most of the states belong to the continuum spectrum
(note again the similarity with the situation for the free PU oscillator with degenerate frequencies  discussed in Sect.3). There are, however, not one, but many bands.
The  bands with positive energies start from $E=\omega$, $E = 2\omega$, etc, and extend upwards while the bands with negative energies  start 
from  $E=-\omega$, $E = -2\omega$, etc, and extend downwards. The energies depend on $N$:
 \be
 \lb{bands}
 E_n(N) \ =\ \frac {\pi n }{2K(k)} \left[ \lambda N \left( 4 + \frac {\omega^4}{\lambda N}
 \right) \right]^{1/4}
  \ee
  with $n = \pm 1, \pm 2, \ldots$ and $K(k)$  defined in \p{period}.
  This means that the levels with $E \in (\omega, 2\omega)$ and $E \in (-2\omega, 
  -\omega)$ are not degenerate, the levels with  $E \in (2\omega, 3\omega)$ and $E \in (-3\omega, 
  -2\omega)$ are doubly degenerate, the levels  with  $E \in (3\omega, 4\omega)$ and $E \in (-4\omega, 
  -3\omega)$ are 3-fold  degenerate, etc. 
   The wave functions are 
     \be
     \lb{wave-fun-cont}
     \Psi_{nN}(x,D)  \propto \frac 1 {\sqrt{N - V(x)}} \exp \left\{
    i D\sqrt{2[N - V(x)]} + \frac {iE_n(N)}{\sqrt{2}} \int^x \frac {dy} {\sqrt{N - V(y)}} \right\}\, .
       \ee

If setting $n=0$ in Eq.\p{bands} and substituting this in \p{wave-fun-cont}, we obtain an infinity of zero-energy states with the wave functions
   \be
   \lb{Psi0N}
   \Psi_{0N}(x,D) \propto \frac 1 {\sqrt{N - V(x)}} \exp \left\{
    i D\sqrt{2[N - V(x)]}\right\}\, .
      \ee
      The functions \p{Psi0N} are not normalizable, but we can in fact choose  the basis with normalizable eigenfunctions. Indeed, one can show that any function
      \be
   \lb{Psi0g}
   \Psi_0^{(g)}(x,D)  \ =\ \int_{-\infty}^\infty g\left( \frac {P^2}2 + V(x) \right)
   e^{iPD} dP  
      \ee
      is a solution to the Schr\"odinger equation $\hat{H} \Psi = 0$. The solutions \p{Psi0N}
      are obtained from \p{Psi0g} by setting $g(N) = \delta(N - N_0)$. But we can also choose $g_k(N) = N^k e^{-N}$ giving the normalizable functions $\Psi_{0k}$. The full spectrum of the quantum Hamiltonian is represented in Fig. \ref{spectr-Robert}. It is Hermitian and the evolution operator is unitary.

    \begin{figure}[h]
   \begin{center}
 \includegraphics[width=5 in]{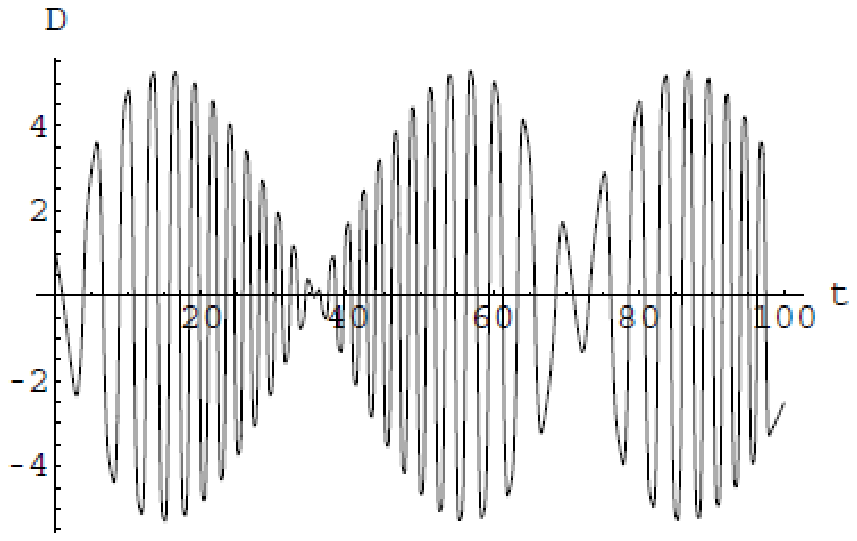}
    \end{center}
\caption{A typical trajectory for the Hamiltonian \p{H-gamma}.} 
\label{bienie}
\end{figure}

\begin{figure}[h]
   \begin{center}
 \includegraphics[width=4 in]{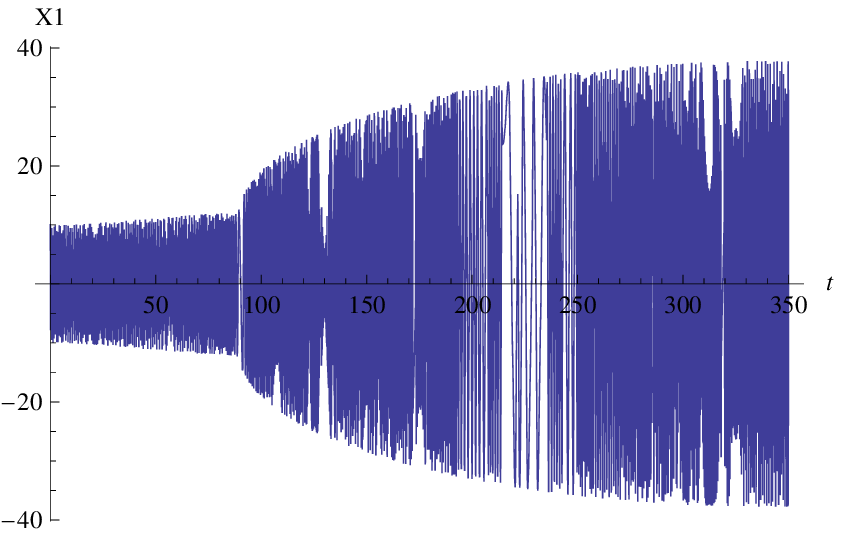}
    \end{center}
\caption{ A typical trajectory for the Hamiltonian  \p{H-raznost-nonlin-12}.} 
  \label{nevyr}
  \end{figure}

 Several other models with benign ghosts are now known. They all represent modifications  either of \p{H-Robert} or of \p{H-raznost} with some particular nonlinear interaction terms added.
  
 In Ref.\cite{Robert} a model with the Hamiltonian
  \be
   \lb{H-gamma}
   H \ =\ pP + DV'(x) - \frac \gamma 2 (D^2 + P^2) 
     \ee
  was considered [it follows from the supersymmetric quantum mechanical model with the action \p{HD-SQM}]. This model is not integrable and can only be studied numerically. This numerical study exhibits a benign behaviour of classical trajectories.
  The latter are even more handy than for the model \p{H-Robert} --- if $\gamma$ is not too large, the linear growth of the amplitudes
  of the oscillations
  is absent. \footnote{A numerical analysis shows that, if $\gamma$ exceeds some limit that depends on the values of $\omega$, $\lambda$ and of the initial conditions, the amplitude  starts to grow again, and this growth is exponential.} One observes instead a finite motion with beats (Fig. \ref{bienie}).

  Alternatively, one can modify the Hamiltonian \p{H-raznost-Robert} by allowing for 
  the frequencies in the ordinary and ghost sector to be different \cite{PU-Sigma}:
    \be
      \lb{H-raznost-nonlin-12}
      H \ =\ \frac{P_1^2 + \omega_1^2 X_1^2}2 - 
      \frac{P_2^2 + \omega_2^2 X_2^2}2 + \kappa (X_1 - X_2) (X_1 + X_2)^3 \, .
       \ee
       Also in this case the motion is finite for not too large $\kappa$. A typical trajectory is displayed in Fig. \ref{nevyr}.
        The system  \p{H-raznost-nonlin-12} behaves thus basically in the same way as 
  the nondegenerate PU oscillator [which is not surprising, bearing in mind that the Hamiltonian of the latter is equivalent to the quadratic part of the Hamiltonian  \p{H-raznost-nonlin-12}].  The fact that the classical motion is finite tells us that the quantum spectrum 
  should be pure point, involving only normalizable states. \footnote{On the other hand, the classical motion for the Hamiltonian \p{H-Robert} is infinite and that is the reason why the spectrum  in Fig. \ref{spectr-Robert} involves the continuum bands.}
   
   The same should hold for the
  Hamiltonian \p{H-gamma}. In fact, the models \p{H-gamma} and \p{H-raznost-nonlin-12} belong to the same family. The second term in \p{H-gamma} modifies the quadratic part of the Hamiltonian. By a proper canonical transformation, this quadratic part can  be brought into the form \p{H-raznost}; the interaction term would acquire then a certain complicated form. 
  
  The findings of \cite{Carroll,Robert,malicious,PU-Sigma} were confirmed in \cite{Pavsic,Ilhan}, where some other nonlinear modifications of \p{H-raznost} bringing about islands of stability, as well as the modifications leaving {\it all} the trajectories stable were suggested. The latter happens e.g. for the interactions
   \be
   V(X_1,X_2) \ =\ \lambda \sin^4(X_1+X_2)
    \ee
   (the suggestion of \cite{Pavsic}) or for
    \be
   V(X_1,X_2) \ =\  \lambda(X_1^4 - X_2^4) + \mu X_1^2 X_2^2 \, , \ \ \ \ \ \ \lambda \gg \mu
     \ee    
  (the suggestion of \cite{Ilhan}).

\section{Benign ghosts: field theory}
\setcounter{equation}0

We know today only one example of a field theory enjoying benign ghosts \cite{duhi-v-pole}. It is a straighforward generalization of \p{H-Robert}. The Lagrangian of the model
is \footnote{It is the bosonic part of the supersymmetric Lagrangian \p{act2D}.}
  \be
  \lb{L-pole}
{\cal L} \ =\ \partial_\mu \phi \partial_\mu D - D V'(\phi) \, ,
  \ee
  where $\mu = 0,1$ (the model is 2-dimensional) and 
  \be
\lb{Vphi}
 V(\phi) = \ \frac {\omega^2  \phi^2}2 + \frac {\lambda \phi^4}4 \, , \ \ \ \ \ \lambda > 0 \, .
 \ee
 The model has two integrals of motion: the energy 
  \be
\lb{E-pole}
E \ =\ \int dx \left[ \dot{\phi} \dot{D} + \partial_x \phi \, \partial_x D + D\phi(\omega^2 + \lambda \phi^2) \right]\,,
\ee
which can be both positive and negative, and the positive definite
\be
\lb{N-pole} 
N \ =\ \int dx \left\{ \frac 12 \left[ \dot{\phi}^2 + (\partial_x \phi)^2 \right] + \frac {\omega^2 \phi^2}2
+ \frac {\lambda \phi^4}4 \right\} \, .
 \ee
  With only 
two integrals of motion for the infinite number of degrees of freedom, the model is not integrable and can only be solved numerically. We did so in the case of only one spatial dimension. The equations of motion are
 \be
\lb{eqmot-pole}
\Box \phi + \omega^2 \phi + \lambda \phi^3 &=& 0 \nn
\Box D + D(\omega^2 + 3\lambda \phi^2) &=& 0 \, .
 \ee
 We see that $\phi(x,t)$ satisfies a nonlinear wave equation. The solutions to this equation cannot grow --- a growth is incompatible with the conservation of $N$.
 The amplitude of the oscillations of the field $D(x,t)$ can grow with time, however. 
\begin{figure}[ht!]

     \begin{center}

            \includegraphics[width=0.6\textwidth]{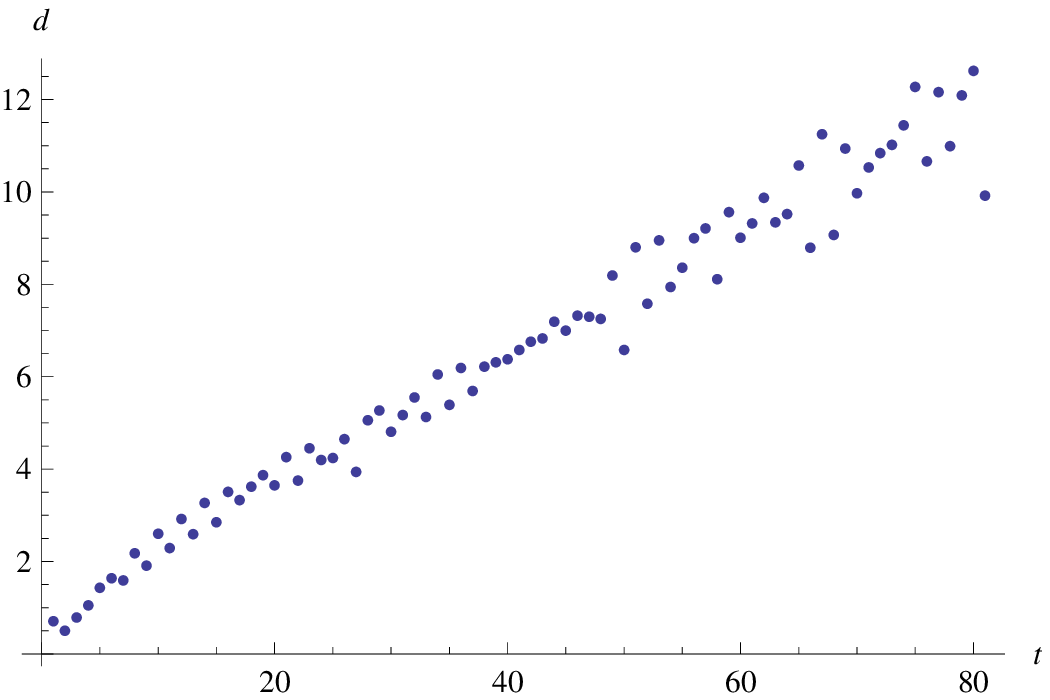}
        
    \end{center}
    \caption{Dispersion $d = \sqrt{\langle D^2 \rangle}$ as a function of time
with the initial conditions \p{incon}. }
\lb{Drost-pole}
    \end{figure}

We played with different values of the parameters $\omega, \lambda$ and with different initial conditions and never found a collapse, only  at worst  a linear growth of
$D(x,t)$  with time.
A typical behaviour is shown in Fig. \ref{Drost-pole}, where the dispersion $d(t) = \sqrt{\langle D^2 \rangle_{x}}$ is plotted as a function of time.  This particular graph corresponds to the choice 
of the parameters $\omega=\lambda = 1$, to the length of the spatial box $L = 20$
(periodic boundary conditions in the spatial direction were imposed),  while the initial conditions at $t=0$ were chosen as follows:
    \be
\lb{incon}
\phi(x,0) = 5e^{-x^2},   \ \ \ \ \  D(x,0)    = \cos (\pi x/20)  \,, \nn
\dot{\phi}(x,0) \ = \ \dot{D}(x,0) \ =\ 0 \, .
 \ee

    The dispersion undergoes stochastic fluctuations, as is 
    natural for a non-integrable nonlinear system.  However, on the average, $d(t)$
    grows linearly with time, similar to Fig. \ref{Drost}, and there is no trace of collapse.

\section{Discussion}
\setcounter{equation}0

Our main message is that the presence of ghosts (unboundness of the spectrum of the Hamiltonian) does not necessarily mean a collapse and disaster --- there are nontrivial
interacting ghost-ridden systems where this does not happen. In spite of the presence of ghosts, HD theories should be treated seriously, one should not throw them immediately 
to a garbage busket, as it was habitually done for many years. Instead, they should be interrogated one by one: whether a HD theory is benign or not (involves or does not involve the collapse) is a nontrivial dynamical question.

Unfortunately, we do not know today a lot of such  benign ghost theories, and our dream to find a HD theory which might play the role of the TOE is far from being fulfilled. All the mechanical examples considered in Sect. 6 boil down to the Hamiltonian \p{H-raznost},
supplemented by different interaction terms. These systems involve ghosts, but they are not HD systems the way they are written.
 It is true that the  Hamiltonian \p{H-raznost} is equivalent to the Hamiltonian \p{HPU} of a HD system --- the Pais-Uhlenbeck oscillator
with different frequencies. However, the nonlinear term like in \p{H-raznost-nonlin-12}
expressed in terms of the original phase space variables $(x,v, p_x, p_v)$ acquires a complicated form, and the corresponding modification of the Lagrangian \p{LPU} is still more complicated and highly nonlocal.

Simple modifications of \p{LPU} do not work. We have seen in Sect. 5 that a modified system like \p{LPU-nonlin} can have an island of stability, but this island is encircled by the sea of unstable collapsing trajectories. Most HD theories that one may try, e.g. the
6-dimensional Lagrangian \p{LD=6} and its supersymmetric generalization \cite{IVZ}, do not even have  islands of stability, only the high sea of collapse.

The only known benign field-theory system was presented and discussed in Sect. 7.
It lives in one spatial and one temporal dimension. Its classical trajectories with the initial conditions that we tried do not collapse. \footnote{One can {\it conjecture} that it is true for all the trajectories and we do so, but it is very difficult, of course, to {\it prove} this conjecture.}
But what is the quantum dynamics of the model \p{L-pole}?

This is a hard question. Ordinary field theories may be studied perturbatively. We single out the free Hamiltonian $H_0$, find its spectrum and treat nonlinear terms as perturbation.
If the spectrum of $H_0$ includes a vacuum state and particle excitations above it, one may define the scattering matrix whose elements represent transition amplitudes between the particle asymptotic states. These transition amplitudes can in many cases be found perturbatively.   

And for the system \p{L-pole}, the spectrum does not have a familiar particle interpretation. The latter works well for an ordinary field theory. When the interaction
is switched off, the Hamiltonian represents a sum of ordinary oscillator Hamiltonians for each Fourier mode. We have a conventional Fock space: the vacuum; the first excited
oscillator level corresponding to a single particle, the second level correspoding to two particle, etc.

Here this approach does not work. If the interaction is switched off, each level involves an
infinite degeneracy: an ``ordinary particle'' (an excitation of the positive-energy oscillator) has the same energy as the state with 2 ordinary particles and one ghost, 3 particles and 2 ghosts, etc. And when the interaction is present ($\lambda \neq 0$), the spectrum is radically reshuffled. If we single out a particular Fourier mode and suppress artificially its interaction with the other modes, the Hamiltonian in this sector looks similar to \p{H-raznost-Robert} and its spectrum includes continuum bands as in Fig. \ref{spectr-Robert}. It is rather clear that even if $\lambda$ is small, the spectrum of the full Hamiltonian for the system \p{L-pole} has little to do with the spectrum of the free system.
\footnote{Also in ordinary  QCD-like theories  without ghosts but with confinement the spectrum of the full theory has nothing to do with the spectrum of $H_0$. Still, in those cases {\it (i)} there is a kinematical region where the perturbative spectrum and perturbative scattering amplitudes make sense; {\it (ii)} even though the spectra of excitations for the full Hamiltonian $H$ and for the free Hamiltonian $H_0$ are very different, the ground vacuum state is well defined in both cases and the physical spectrum involves only ordinary particles carrying positive energies. It is difficult to {\it calculate} the hadron scattering amplitudes, but they can be {\it defined} and {\it measured}. } If trying to describe it in terms of particles, these particles (both the ordinary particles and the ghosts) must have a continuum 
spectrum of {\it masses}. \footnote{One can recall at this point supermembranes, where the mass spectrum is also continuous \cite{memb}.} 

Without any doubt, all this looks very strange and unusual. But strange and unusual does not mean self-contradictory and inconsistent. 
    And one can be almost sure that the TOE, whatever it is, must be strange and unusual. After all, people tried many less strange and more usual candidates for the role of the TOE, including strings, but these efforts have not been successful so far. 
    
    We have just recalled confining theories, where the spectra of the free Hamiltonian and of the full Hamiltonian are rather different. 
    Bearing this in mind, the scenario of  {\it confinement of ghosts} was discussed in the literature. The main idea is that even if the ghosts are present in the bare Hamiltonian, they may disappear from the physical spectrum due to confinement \cite{spon,Donoghue}. This conjecture was inspired by the hypothesis \cite{spon,Zee} that the fundamental gravity action is not the Einstein-Hilbert action, but the conformally invariant
    Weyl HD action
     \be
     \lb{Weyl-act}
     S^{\rm Weyl} \ \sim \ \frac 1{f^2}\int \sqrt{-g}d^4x \,  C_{\mu\nu\lambda\sigma}  C^{\mu\nu\lambda\sigma} \, ,
      \ee
      where
      \be
      \lb{Weyl-tens}
       C_{\mu\nu\lambda\sigma} \ =\  R_{\mu\nu\lambda\sigma} - \frac 12 (g_{\mu\lambda} R_{\nu\sigma} -
    g_{\mu\sigma} R_{\nu\lambda} - g_{\nu\lambda} R_{\mu\sigma} +  g_{\nu\sigma}  R_{\mu\lambda} )
    + \frac 16 (  g_{\mu\lambda} g_{\nu\sigma} - g_{\mu\sigma} g_{\nu\lambda}) 
      \ee
      is the Weyl tensor, or a supersymmetric generalization of this action.  Conformal gravity and conformal supergravities are asymptotically free \cite{Fradkin} so that classical (super)conformal symmetry is broken there by quantum effects. Like in QCD, this brings about a dimensional parameter --- the scale where the coupling constant $f$ becomes large. It is natural to assume that this scale is of order of the Planck mass $M_{Pl}$.  The EH term in the effective low-energy action along with the cosmological constant term are {\it induced} due to quantum effects by Sakharov's mechanism \cite{Sakharov,Adler}. \footnote{Nobody has an idea today how to get rid of the huge cosmological constant, which appears in all approaches. This is a separate very hard question.}
   
   The propagator of the metric following from \p{Weyl-act} behaves as
     \be
      \lb{D(k)}
      D(k) \ \sim \ \frac {f^2}{k^4} \, ,
        \ee
        which brings about the confining static potential
        \be
        \lb{pot-stat}
        V(r) \ \propto \ f^2 \int \frac {d^3 k}{k^4}\, e^{i {\small \vecg{k}} {\small \vecg{r}}} \ \propto f^2 r \, .
          \ee
          One may thus suppose that this potential confines the ghost degrees of freedom and only ordinary excitations are left in the physical spectrum.
           
           \vspace{1mm}
           
    However, 
    \begin{itemize}
    \item We now do not think that {\it any} gravity theory, regardless of whether it is the conventional Einstein's gravity or a HD gravity,
     is a part of the TOE. We argued in the Introduction that the TOE should rather be a HD theory living in a flat higher-dimensional bulk, while
    the gravity is an effective theory living on the 3-brane, which is our Universe. 
    \item As was mentioned above, the physical spectrum 
    can be very different from the spectrum of the free Hamiltonian, and many degrees of freedom that are seen in the latter
    may indeed disappear from the spectrum --- be confined. But a general theorem proven at the end of Sect. 3 and valid for {\it any} nondegenerate HD system dictates that  
    the physical spectrum of its Hamiltonian does not have a ground state and must involve ghosts. Thus,  confinement or
    any other mechanism can drastically affect the spectrum of a HD Hamiltonian, but it does not allow one to get rid of
    all the ghosts. 
    \end{itemize}

We are indebted to D. Robert and R. Woodard for illuminating discussions.

     \section*{Appendix A: Jordan blocks, unitarity and continuum spectrum.}
    \setcounter{equation}0
    \def\theequation{A.\arabic{equation}}
    We will explain here, following Refs. \cite{PU-Sigma,except}, how the Hamiltonian \p{HPU-degen} of the degenerate Pais-Uhlenbeck oscillator may be expressed as a set of Jordan blocks of infinite dimension, and why such representation does not lead to the loss of unitarity and only signalizes the presence of continuum spectrum.
    
    Consider first the nondegenerate Hamiltonian \p{HPU}. It has the spectrum \p{spec-PU}. As was mentioned, this spectrum is pure point: the eigenfunctions are normalizable. To find them explicitly, it is convenient to represent the wave function as follows \cite{duhiPL}:
    \be
 \label{Psi-nondegen}
\Psi_{nm}(x,v)\ \sim \ e^{-i\omega_1 \omega_2 xv} \exp\left\{ - \frac \Delta 2 \left( v^2 + 
\Omega_1 \Omega_2 x^2 \right) \right\} \phi_{nm}(x,v)\ ,
 \ee
where $\Delta = \omega_1 - \omega_2 >0$. The Hamiltonian acting on the function $\phi_{nm}$ reads
 \be
 \lb{Hamphi}
 \hat{H}_\phi \ =\ - \frac 12 \frac {\partial^2}{\partial v^2} + (\Delta v + i\omega_1 \omega_2 x) \frac \partial {\partial v} 
 -iv \frac \partial {\partial x} + \frac \Delta 2 \, .
   \ee
   We now introduce the variables 
    \be
    \lb{zu}
    z \ =\ \omega_1 x + iv, \ \ \ \ \ \ \ \ \ \ \ \ u \ =\ \omega_2 x - iv \, ,
     \ee
     after which the operator \p{Hamphi}  acquires the form
     \be
 \lb{Hamphi-zu}
 \hat{H}_\phi(z,u) \ =\  \frac 12 \left( \frac {\partial}{\partial z} - \frac {\partial}{\partial u} \right)^2 + 
 \omega_1 u \frac {\partial}{\partial u}  -  \omega_2 z \frac {\partial}{\partial z} 
  + \frac \Delta 2 \, .
   \ee
   The holomorphicity of $ \hat{H}_\phi(z,u)$ means that its eigenstates are holomorphic functions $\phi(z,u)$. 
   An obvious eigenfunction with the eigenvalue $\Delta/2$ is $\phi(z,u) = const$. Further, if we assume $\phi$ to be a function
   of only one holomorphic variable $u$ or $z$, the equation $\hat{H}_\phi \phi = E\phi$ acquires the same form as for the
   equation for the preexponential factor in the standard oscillator problem. Its solutions are Hermite polynomials
    \be
    \lb{Hermit}
\phi_n(u) \ =\  \ H_n[i\sqrt{\omega_1} u] \equiv H_n^+ \, , \ \ \ \ \ \ \ E_n  \ =\ \frac \Delta 2 + n \omega_1 \, , \nonumber \\
\phi_m(z) \ =\  \ H_m[\sqrt{\omega_2} z] \equiv H_n^- \, , \ \ \ \ \ \ \ E_m  \ =\ \frac \Delta 2 - m \omega_2 
 \ee
 [$H_0(w) \equiv 1, H_1(w) \equiv  2w, H_2(w) \equiv  4w^2-2, \ldots$].
 The solutions \p{Hermit} correspond to excitations of only one of the oscillators. For sure, there are also the states where both oscillators are excited. One can be directly convinced that the functions 
 \be
\label{sumHerm}
\phi_{nm}(x, v) &=& \sum_{k=0}^m \left( \frac {i\Delta }{4 \sqrt{\omega_1 \omega_2}} \right)^k 
\frac {m! (n-m)!}{(m-k)! k! (n-m+k)!} H^+_{n-m+k} H^-_k,\ \ \ \  n\geq m \ , \nonumber \\
\phi_{nm}(x, v) &=& \sum_{k=0}^n \left( \frac {i\Delta }{4 \sqrt{\omega_1 \omega_2}} \right)^k 
\frac {n! (m-n)! }{(n-k)! k! (m-n+k)!} H^+_{k} H^-_{m-n+k},\ \ \ \  m  \geq  n\, .
 \ee
 are the eigenfunctions of the operator \p{Hamphi-zu} with the eigenvalues \p{spec-PU} \cite{PU-Sigma}. Multiplying the polynomials \p{sumHerm}
 by the exponential factor as dictated by Eq.\p{Psi-nondegen}, we arrive at the normalizable wave functions for the Hamiltonian \p{HPU}. 
    
    Now let us see what happens in the limit $\Delta \to 0$. Of course, this limit is singular: the canonical transformations \p{XP-canon} are not defined at $\Delta = 0$. But it is instructive to follow the fate of the eigenfunctions \p{Psi-nondegen} in this limit.
    
    To begin with, we observe that the functions \p{Psi-nondegen} are not normalizable any more. This already indicates that the spectrum is no longer pure point, but involves continuum states.  Next we see that only the first term with $k=0$ is left in the sums \p{sumHerm}.
    We note moreover that this first term does not depend on $n$ and $m$ separately, but only on the difference $n-m$ ! We derive 
    \be
\label{Psi-degen}
 \Psi_{nm}(x,v)\ \propto \ \left[ 
\begin{array}{c}  e^{-i\omega^2 xv} H^+_{n-m},\ \ \ \ \ m \leq n \\
e^{-i\omega^2 xv} H^-_{m-n},\ \ \ \ \ m > n  \end{array} \right. .
 \ee
    
    A point in the space of parameters where different eigenfunctions of the Hamiltonian merge together is called the {\it exceptional point} \cite{Heiss}. The simplest example of such a point is given by the matrix Hamiltonian 
     \be
\label{H-Jordan}
\hat{H} = \left( \begin{array}{cc}1&1 \\ \Delta&1  \end{array} \right) 
 \ee
  If $\Delta > 0$, the matrix \p{H-Jordan} has two different real eigenvalues $\lambda_{1,2} = 1 \pm \sqrt{\Delta}$. 
  The eigenvectors are also different. But in the limit $\Delta \to 0$, these two eigenvectors merge into one,  $\psi = \left( \begin{array}{c}1 \\ 0 \end{array} \right) $. The Hilbert space where the Hamiltonian \p{H-Jordan} is allowed to  act has thus shrinked: it is now one-dimensional, rather then two-dimensional.
  And if one insists on working in the original 2-dimensional Hilbert space, the evolution is no longer unitary. Indeed, a generic solution of the Schr\"odinger equation 
     \be
\label{Schr}
i \frac {d\Psi}{dt} \ = H\Psi
\ee
is in this case 
\be
\label{solab}
\Psi(t) \ =\ a \left( \begin{array}{c}1 \\ 0 \end{array} \right)e^{-it} + 
b \left( \begin{array}{c} -it \\ 1 \end{array} \right)e^{-it} \ .
 \ee
 When $b \neq 0$, the norm of $\Psi(t)$ grows with time.
 
 In the limit $\Delta \to 0$, the matrix \p{H-Jordan} acquires a Jordan form. The specifics of the Hamiltonian \p{HPU} is, however, that the exceptional point $\omega_1 = \omega_2$ has an {\it infinite order}. It is already clear from 
 comparison \p{Psi-nondegen} and \p{Psi-degen} --- an infinite number of eigenstates are coalesced into one at $\Delta = 0$. To see the appearance of a Jordan block quite explcitly, consider for example the sector with $n-m = 0$ and
 choose the basis 
  \be
  \lb{basis}
  \Psi_k \ =\ C_k e^{-i\omega^2 xv} H_k[\sqrt{\omega}(v+i\omega x)] \, H_k[\sqrt{\omega}(\omega x + iv)] \equiv 
  C_k e^{-i\omega^2 xv} H_k^+ H_k^-
     \ee
      Now, $\Psi_0$ is an eigenstate of the Hamiltonian \p{HPU-degen} with zero energy. Using the properties of the Hermite polynomials, it is not  difficult to derive
      \footnote{In the sectors with $n-m = p \neq 0$, we can choose the basis $\Psi_k \propto H^+_{p+k} H^-_k$  for positive $p$ and the
    basis  $\Psi_k \propto H^+_k H^-_{-p+k}$ for negative $p$ and obtain a  relation similar to \p{H-na-Psik} with the extra term $\omega p \Psi_k$ on the right-hand side. }
       \be
       \lb{H-na-Psik}
       \hat{H}\Psi_k \ =\ -4ik^2 \frac {C_k}{C_{k-1}}\omega \, \Psi_{k-1} \, .
         \ee  
        
  Representing it in the matrix form and choosing the coefficients $C_k$ appropriately,  we obtain a Jordan block of infinite dimension:
  \be
  \lb{inf-Jordan}
  \hat{H} \ =\ \left( \begin{array}{ccccc}0&1&\cdots&\cdots &\cdots \\ 0&0&1& \cdots& \cdots \\ 0 & 0 & 0 & 1& \cdots \\  \cdots & \cdots & \cdots& \cdots & \cdots  \end{array} \right) \, .
   \ee

 The physics in this case is completely different compared to the case of a finite Jordan block. If the matrix \p{inf-Jordan} were truncated at any finite order, it would have only {\it one} eigenvector. But the infinite matrix \p{inf-Jordan} has an infinity of eigenvectors 
   \be
    \lb{vect-eps}
    \Psi_\epsilon \ =\  \left(\begin{array}{c} 1\\ \epsilon \\ \epsilon^2 \\ \cdots \end{array} \right)
     \ee
     with eigenvalues $\epsilon$. When time evolves, the vector \p{vect-eps} is multiplied by $e^{-it \epsilon}$ and its norm is not changed. As anticipated, our Hamiltonian has a continuous spectrum, but there are no pathologies: unitarity is preserved.  
     
     To clarify the relationship between the approach of Sect. 3, where the solutions \p{Psi-state} were obtained by a direct
     analysis of the Hamiltonian \p{HPU-degen}, and the approach of this Appendix, where the spectral problem for the degenerate Hamiltonian \p{HPU-degen} was treated as a limit of a spectral problem for the nondegenerate Hamiltonian, we expand the Bessel function in the solution \p{Psi-state} in the Taylor series and use for each term of this expansion the identity 
    \be \exp\left\{ - \frac 14 \frac {\partial^2}{\partial z^2} \right\} \, z^n
\ =\ 2^{-n} H_n(z) \,.
 \ee
 We thus derive (for positive $l$):
 \be
\label{Psilk}
\Psi_{lk} \ \propto \ e^{-it(l\omega + k^2/4)} 
e^{-i\omega^2 xv} \sum_{m=0}^\infty \left( \frac {i k^2}\omega \right)^m 
\frac {H^+_{l+m} H^-_m}{2^{4m} m!(l+m)!}\ . 
 \ee
 And this infinite series over the basis $\propto H^+_{l+m}H^-_m$ has the same meaning as the infinite column \p{vect-eps}.
 
    \section*{Appendix B: Supersymmetry and ghosts}
    \setcounter{equation}0
    \def\theequation{B.\arabic{equation}}
 
    We will elucidate here the genetic relationship of the ghost Hamiltonians \p{H-Robert}
    and \p{H-gamma} and the field theory system \p{L-pole} to certain higher-derivative supersymmetric actions.
    
    Introduce one-dimensional superspace $(t; \theta, \bar\theta)$, where $\theta$ is a complex Grassmann superpartner of time, and consider a real supervariable
      \be
      \lb{super-X}
      X(t; \theta, \bar\theta) \ =\ x(t) + \theta \bar \psi(t) + \psi(t) \bar \theta + D(t) \theta \bar\theta \, .
       \ee
       The action of the simplest supersymmetric mechanical model reads \cite{WitSQM}
\be
\label{Wit-mod}
S^{\rm Witten} \ = \ \int dt d\bar \theta d\theta \left[ \frac 12 ( \bar {\cal D} X)  ({\cal D} X) -
V(X) \right] \, ,
 \ee
 where $V(X)$ is an arbitrary function (called {\it superpotential}) and 
 \be
 \lb{covdev} 
{\cal D} = \frac \partial {\partial \theta} + i\bar \theta \frac \partial {\partial t}, \ \ \ 
\bar  {\cal D} = -\frac \partial {\partial \bar \theta} - i\theta \frac \partial {\partial t}
  \ee
are the supersymmetric covariant derivatives. Substituting \p{super-X} into \p{Wit-mod}, we derive the component Lagrangian

\be
\lb{Wit-mod-comp}
L \ =\ \frac {\dot x^2 + D^2}2 + \frac i2\left(\dot \psi \bar \psi - \psi \dot{\bar \psi} \right) -  V'(x)D  - V''(x)\bar \psi \psi  \, .
 \ee
 
 The variable $D$ enters this expression without derivatives. It thus represents an auxiliary non-dynamic variable and can be algebraically excluded. 
 
 Consider, however, the action involving an extra time derivative:
   \be
\label{S-HD-SQM}
S^{\rm HD} \ = \ \int dt d\bar \theta d\theta \left[\frac i2 ( \bar {\cal D} X) \frac d{dt} ({\cal D} X) -
V(X) \right] \, .
 \ee
The component Lagrangian then reads
\be
\label{L-HD-SQM}
L = \dot x \dot D - V'(x) D - V''(x)\bar \psi \psi + 
\dot {\bar \psi} \dot \psi \ .
 \ee
Now the time derivative $\dot{D}(t)$ enters the Lagrangian, and $D(t)$ is a genuine dynamical variable. Suppressing in \p{L-HD-SQM}  the fermion part and performing the Legendre transformation, we derive the Hamiltonian \p{H-Robert}.

The Hamiltonian \p{H-gamma} is derived in the same way from the action of the mixed
model, including the higher-derivative term and Witten's term:
\be
\label{HD-SQM-mixed}
S^{\rm mixed} \ = \ \int dt d\bar \theta d\theta \left[ \frac i2 ( \bar {\cal D} X) \frac d{dt} ({\cal D} X) + \frac \gamma 2 ( \bar {\cal D} X) ({\cal D} X) -
V(X) \right] \, .
 \ee

To derive the field theory model \p{L-pole}, we introduce a real $2D$ superfield 
$\Phi(x,t; \theta, \bar\theta)$ and write the action as 

\be
\lb{act2D}
S \ =\ \int dt dx d\bar\theta d\theta  \left[ -2i {\cal D} \Phi \partial_+ 
{\cal D} \Phi - V(\Phi) \right] \, ,
 \ee
where $\partial_\pm = (\partial_t \pm \partial_x)/2$ and 
 \be
\lb{DDbar}
{\cal D} = \frac {\partial}{\partial \theta} + i\theta \partial_-, \ \ \ \ \ \ \ \ \ \ 
\bar {\cal D} = \frac {\partial}{\partial \bar\theta} - i \bar\theta \partial_+
 \ee
are the $2D$ supersymmetric covariant derivatives.
\footnote{See the textbook \cite{West} for a good pedagogical description of the $2D$ superfield formalism. Note that the first term in \p{act2D} can as well be written as 
$-2i \bar {\cal D} \Phi \partial_- \bar {\cal D} \Phi$. The integrals over $ d\bar\theta d\theta$ of these two
expressions coincide. }
 The component Lagrangian reads 
 \be
\lb{L2D}
{\cal L} = \partial_\mu \phi  \, \partial_\mu D \, + \, \partial_\mu \bar\psi  \, \partial_\mu \psi
-  DV'(\phi) \, - \, V''(\phi) \bar \psi \psi \, .
 \ee
Suppressing the fermion part, we arrive at \p{L-pole}.

\end{document}